\documentclass[11pt]{article}
\usepackage{amssymb}
\usepackage[perpage,symbol]{footmisc}
\usepackage{listings}
\usepackage{amsfonts}
\usepackage{enumerate}
\usepackage{amsmath}
\usepackage{cite}
\usepackage{array}
\usepackage{booktabs}
\newcommand\relphantom[1]{\mathrel{\phantom{#1}}}

\newcommand{\tabincell}[2]{\begin{tabular}{@{}#1@{}}#2\end{tabular}}

\begin{document}

\newtheorem{theorem}{Theorem}[section]
\newtheorem{corollary}[theorem]{Corollary}
\newtheorem{definition}[theorem]{Definition}
\newtheorem{proposition}[theorem]{Proposition}
\newtheorem{lemma}[theorem]{Lemma}
\newtheorem{example}[theorem]{Example}
\newenvironment{proof}{\noindent {\bf Proof.}}{\rule{3mm}{3mm}\par\medskip}
\newcommand{\remark}{\medskip\par\noindent {\bf Remark.~~}}

\title{A class of cyclic codes whose dual have five zeros}
\author{Yan Liu\footnote{Corresponding author, Dept. of Math., SJTU, Shanghai, 200240,  liuyan0916@sjtu.edu.cn.}, Chunlei Liu\footnote{Dept. of Math., Shanghai Jiaotong Univ., Shanghai,
200240, clliu@sjtu.edu.cn.}}

\date{}
\maketitle
\thispagestyle{empty}

\abstract{In this paper, a family of  cyclic codes over $\mathbb{F}_{p}$  whose duals have five zeros is presented, where  $p$ is  an odd prime. Furthermore, the weight distributions of these cyclic codes are determined. }



\section{\small{INTRODUCTION}}
Recall that an $[n,l,d]$ linear code $\mathcal{C}$ over $\mathbb{F}_{p}$ is a linear subspace of $\mathbb{F}_{p}^{n}$ with dimension $l$ and minimum Hamming distance $d$. Let $A_{i}$ denote the number of codewords  in $\mathcal{C}$ with Hamming weight $i$.
%The   weight enumerator of $\mathcal{C}$ is defined by $1+A_{1}Z+A_{2}Z^{2}+\cdots+A_{n}Z^{n}$.
The sequence $(A_{0}, A_{1}, A_{2},\ldots, A_{n})$ is called the weight distribution of the code $\mathcal{C}$.
And $\mathcal{C}$ is called cyclic if for any $(c_{0}, c_{1},  \ldots, c_{n-1}) \in \mathcal{C}$, also $(c_{n-1}, c_{0},  \ldots, c_{n-2}) \in \mathcal{C}$.
%By identifying any vector $(c_{0}, c_{1},  \ldots, c_{n-1}) \in \mathbb{F}_{p}^{n}$  with $c_{0}+ c_{1}x+\cdots+c_{n-1}x^{n-1} \in \mathbb{F}_{p}[x]/(x^{n}-1)$, any cyclic code corresponds to an ideal of the
A linear code $\mathcal{C}$ in $\mathbb{F}_{p}^{n}$ is cyclic if and only if $\mathcal{C}$ is an ideal of the polynomial residue class ring $\mathbb{F}_{p}[x]/(x^{n}-1)$. Since $\mathbb{F}_{p}[x]/(x^{n}-1)$ is a principal ideal ring, every cyclic code corresponds to a principal ideal $(g(x))$ of  the multiples of a polynomial $g(x)$ which is the monic polynomial of lowest degree in the ideal. This polynomial  $g(x)$ is called the generator polynomial, and $h(x)=(x^{n}-1)/g(x)$ is called the parity-check polynomial of the code $\mathcal{C}$. We also  recall that a cyclic code is called irreducible if its parity-check polynomial  is irreducible over $\mathbb{F}_{p}$ and reducible, otherwise. And a cyclic code over $\mathbb{F}_{p}$ is said to have $t$ zeros if all the zeros of the generator polynomial of the code form $t$ conjugate classes, or equivalently, the generator polynomial has $t$ irreducible factors over $\mathbb{F}_{p}$. Determining the weight distribution of a cyclic code has been the interesting subject of study for many years.
 %
 %Since the weight distribution of a linear code gives its minimum distance and thus the error capability of it.
% Determining the weight distribution of a linear code is an important research object in coding theory.
% For cyclic codes, the error correcting capability may not be as good as with some other linear codes in general. However,  because of  their good algebraic
%structure, the weight distributions of some cyclic codes can be determined by algebraic techniques, exponential sums for example. Besides, cyclic codes have wide applications in storage and %communication systems because they have efficient encoding and decoding algorithms. Moreover, cyclic codes are applied in association schemes \cite{3} and secret schemes \cite{4}.Therefore, %%determining the weight distributions of cyclic codes is not only a problem of theoretical interest, but also of practical importance.
 %have been interesting subjects of study for many years and are very hard problem in general.
For information on the weight distributions of irreducible cyclic codes, the reader is referred to \cite{1,2,5,6}. Information on the weight distributions of reducible cyclic codes could be found in \cite{7,9,10,12,13,14,15,18,19,20,22,28}. In this paper, we will determine the weight distributions of a class of reducible cyclic codes whose duals have five zeros.

Throughout this paper, let $m \geq 5$ be an odd integer and $k$ be a positive integer  such that  gcd$(m,k)=1$. Let $p$ be  an odd prime and $\pi$ be a primitive element of the finite field $\mathbb{F}_{p^{m}}$.

Let $\pi$ be a primitive element of the finite field $\mathbb{F}_{q}$.
Let $h_{0}(x)$, $h_{-0}(x)$, $h_{1}(x)$,$h_{-1}(x)$ and $h_{2}(x)$ be the minimal polynomials of $\pi^{-1}$,$-\pi^{-1}$, $\pi^{-\frac{p^{k}+1}{2}}$, $-\pi^{-\frac{p^{k}+1}{2}}$ and $\pi^{-\frac{p^{2k}+1}{2}}$ over $\mathbb{F}_{p}$, respectively. It is easy to check that $h_{0}(x)$, $h_{-0}(x)$, $h_{1}(x)$,$h_{-1}(x)$ and $h_{2}(x)$ are  polynomials of degree $m$ and are pairwise distinct.  The cyclic codes over  $\mathbb{F}_{p}$ with  parity-check polynomial $h_{0}(x)h_{1}(x)$ have been extensively studied by \cite{4,14,22,23}. Zhengchun Zhou and Cunsheng Ding \cite{19} proved the cyclic codes over  $\mathbb{F}_{p}$ with  parity-check polynomial $h_{-0}(x)h_{1}(x)$ have three nonzero weights and determined its weight distributions. And in \cite{16}, the authors proved  the cyclic codes over  $\mathbb{F}_{p}$ with  parity-check polynomial $h_{0}(x)h_{-0}(x)h_{1}(x)$ have six nonzero weights and determined their weight distribution. Let $\mathcal{C}_{(p,m,k)}$ be the cyclic code with parity-check polynomial $h_{0}(x)h_{-0}(x)h_{1}(x)h_{-1}(x)h_{2}(x)$.  In this paper, we will determine the weight distribution of the cyclic code $\mathcal{C}_{(p,m,k)}$.

From now on, we always assume that $\lambda$ is a fixed non-square element in $\mathbb{F}_{p}$. Since $m$ is odd, it is also a non-square element in $\mathbb{F}_{p^{m}}$. Then  $\lambda x$ runs through all non-square elements of $\mathbb{F}_{p^{m}}$ as $x$ runs through all nonzero square elements of $\mathbb{F}_{p^{m}}$. In addition, we have the following result.
\begin{lemma}[\cite{19}]\label{Le:3.1}
$\lambda^{(1+p^{k})/2}=\lambda$ if $k$ is even, and $\lambda^{(1+p^{k})/2}=-\lambda$ otherwise.
\end{lemma}
The rest of this paper is  organized as follows. Some necessary results on quadratic forms will be introduced in Section 2. In Section 3, we will solve some systems of equations. A family of cyclic codes is presented and their weight distributions are determined in Section 4.

\section{\small{QUADRATIC FORMS OVER FINITE FIELDS}}
We follow the notation in Section 1. In this section, we will recall the definition of the quadratic forms over finite fields and some results about it.  In particular, we present the evaluation of a specific exponential sum  which is derived from the properties of the quadratic forms.  %Quadratic forms have been well studied $($see \cite{11} and the references therein$)$, and have application in design and coding theory.

By identifying $\mathbb{F}_{p^{m}}$ with the $m$-dimensional $\mathbb{F}_{p}$-vector space  $\mathbb{F}_{p}^{m}$,  a function $Q$ from $\mathbb{F}_{p^{m}}$ to $\mathbb{F}_{p}$ can be regarded as an $m$-variable polynomial on $\mathbb{F}_{p}$. Then $Q$ is called a quadratic form over $\mathbb{F}_{p}$ if its corresponding polynomial is a polynomial  of degree two over $\mathbb{F}_{p}$ and  can be represented as
\[Q(x_{1}, x_{2},\ldots, x_{m})=\sum_{1\leq i\leq j\leq m}a_{ij}x_{i}x_{j},
\]
where $a_{ij} \in \mathbb{F}_{p}$.
 The rank of the quadratic form $Q(x)$ is defined as the codimension of the $\mathbb{F}_{p}$-vector space $V=\{x \in \mathbb{F}_{p^{m}}: Q(x+z)-Q(x)-Q(z)=0~ for~ all ~z \in \mathbb{F}_{p^{m}} \}$.

For a quadratic form  $Q(x)$, there exists a symmetric matrix $A$ of order $m$ over $\mathbb{F}_{p}$ such that $Q(x)= XAX^{T}$, where $X=(x_{1}, x_{2},\ldots,x_{m})\in \mathbb{F}_{p}^{m}$ and $X^{T}$ denotes the transpose of $X$. Then there exists a nonsingular matrix $H$ of order $m$ over $\mathbb{F}_{p}$ such that $HAH^{T}$ is a diagonal matrix \cite{11}.
It is easy to check that the rank of the quadratic form  $Q(x)$ is exactly  the rank of $A$.
Under the nonsingular linear substitution $X=ZH$ with $Z=(z_{1}, z_{2},\ldots,z_{m})\in \mathbb{F}_{p}^{m}$,  $Q(x)=ZHAH^{T}Z^{T}=\sum_{i=1}^{r}d_{i}z_{i}^{2}$, where $r$ is the rank of $Q(x)$ and $d_{i} \in \mathbb{F}_{p}^{\ast}$. Let $\triangle= d_{1}d_{2}\cdots d_{r}$ (we assume $\triangle=0$ when $r=0$). We can recall that the Legendre symbol $(\frac{a}{p})$ has the value 1 if $a$ is a quadratic residue mod $p$, $-1$ if $a$ is a quadratic non-residue mod $p$, and zero if $p | a$. Then  $(\frac{\triangle}{p})$ is an invariant of $A$ under the action of $H \in GL_{m}(\mathbb{F}_{p})$. Then we introduce the following two lemmas.
\begin{lemma}[\cite{11}]
With the notation as above, we have
\[
\sum_{x \in \mathbb{F}_{p^{m}}}\zeta_{p}^{Q(x)}=\begin{cases}
           (\frac{\triangle}{p})p^{m-\frac{r}{2}},&  p \equiv 1\pmod4,\\
            (\frac{\triangle}{p})(\sqrt{-1})^{r}p^{m-\frac{r}{2}},& p \equiv 3\pmod4,
           \end{cases}
\]
for any quadratic form $Q(x)$ in $m$ variables of rank $r$ over $\mathbb{F}_{p}$, where $\zeta_{p}$ is a primitive $p$-th root of unity.
\end{lemma}
\begin{lemma} [\cite{19}, Lemma 2.2]\label{Le:2.2}
Let $Q(x)$ be a quadratic form in $m$ variables of rank $r$ over $\mathbb{F}_{p}$, then
 \[
  \sum_{y \in \mathbb{F}_{p}^{\ast}}\sum_{x \in \mathbb{F}_{p^{m}}}\zeta_{p}^{yQ(x)}=\begin{cases}
           \pm (p-1)p^{m-\frac{r}{2}},& \textrm{r even},\\
            0,& otherwise.
           \end{cases}
\]
\end{lemma}
For any fixed $(u,v,w) \in \mathbb{F}_{p^m}^{2}$, let $Q_{u,v,w}(x)=Tr(ux^{2}+vx^{p^{k}+1}+wx^{p^{2k}+1})$,
 %is a $\mathbb{F}_{p}$-linear function over  $\mathbb{F}_{p^{m}}$,
 where $Tr$ is the trace mapping from $\mathbb{F}_{p^{m}}$ to $\mathbb{F}_{p}$. Moreover, we have the following result.
\begin{lemma}[\cite{28}]\label{Le:2.3}
For any $(u,v,w) \in \mathbb{F}_{p^{m}}^{2}\backslash \{(0,0,0)\}$, $Q_{u,v,w}(x)$ is a quadratic form over $\mathbb{F}_{p}$ with rank at least $m-4$.
\end{lemma}
\section{\small{SOLUTIONS OF SOME SYSTEMS OF EQUATIONS }}
In this section, we will solve some systems of equations, which will be needed in the  subsequent section. Before introducing them, for any positive integer $k$, we define $d_{1}=p^{k}+1$ and $d_{2}=p^{2k}+1$.
\begin{lemma}\label{Le:3.9}
Let $N_{2}$ denote the number of solutions $(x_{1},x_{2},y_{1},y_{2}) \in \mathbb{F}_{p^{m}}^{2}\times \mathbb{F}_{p}^{\ast2}$ of the following system of equations
\begin{equation}\label{Eqs:3.1}
\begin{cases}
y_{1}x_{1}^{2}+ y_{2}x_{2}^{2}=0\\
y_{1}x_{1}^{d_{1}}+y_{2}x_{2}^{d_{1}}=0
\end{cases}
\end{equation}
Then $N_{2}=(p-1)^{2}p^{m}$.
\end{lemma}
\begin{proof}
\begin{enumerate}[(1)]
  \item When $y_{1}$ and $y_{2}$ are both squares of $\mathbb{F}_{p}^{\ast}$, the number of the solutions $(x_{1},x_{2},y_{1},y_{2}) \in \mathbb{F}_{p^{m}}^{2}\times \mathbb{F}_{p}^{\ast2}$ of the system above is $\frac{(p-1)^{2}}{4}$ multiplied  by the number of the solutions  $(x_{1},x_{2}) \in \mathbb{F}_{p^{m}}^{2}$ of the following system of equations:
      \begin{equation*}
\begin{cases}
x_{1}^{2}+ x_{2}^{2}=0\\
x_{1}^{d_{1}}+x_{2}^{d_{1}}=0,
\end{cases}
\end{equation*}
which is  equal to the number of the following system of equations:
\begin{equation}\label{Eqs:3.10}
\begin{cases}
x_{1}^{2}+ x_{2}^{2}=0\\
x_{1}^{d_{1}}+ x_{2}^{d_{1}}=0\\
x_{1}^{d_{2}}+ x_{2}^{d_{2}}=0.
\end{cases}
\end{equation}
The number of the the solutions  $(x_{1},x_{2}) \in \mathbb{F}_{p^{m}}^{2}$ of the system (\ref{Eqs:3.10}) is $2p^{m}-1$ when $p \equiv 1 \pmod 4$ and $1$ otherwise by Lemma 3.4 in \cite{28}.
  \item  When $y_{1}$ is a square element but  $y_{2}$ is a non-square element of $\mathbb{F}_{p}^{\ast}$, the number of the solutions $(x_{1},x_{2},y_{1},y_{2}) \in \mathbb{F}_{p^{m}}^{2}\times \mathbb{F}_{p}^{\ast2}$ of the system above is $\frac{(p-1)^{2}}{4}$ multiplied  by the number of the solutions  $(x_{1},x_{2}) \in \mathbb{F}_{p^{m}}^{2}$ of the following system of equations:
      \begin{equation*}
\begin{cases}
x_{1}^{2}+\lambda x_{2}^{2}=0\\
x_{1}^{d_{1}}+\lambda x_{2}^{d_{1}}=0,
\end{cases}
\end{equation*}
which is equal to the number of the following system of equations:
\begin{equation}\label{Eqs:3.11}
\begin{cases}
x_{1}^{2}+\lambda x_{2}^{2}=0\\
x_{1}^{d_{1}}+\lambda x_{2}^{d_{1}}=0\\
x_{1}^{d_{2}}+\lambda x_{2}^{d_{2}}=0.
\end{cases}
\end{equation}
The number of the the solutions  $(x_{1},x_{2}) \in \mathbb{F}_{p^{m}}^{2}$ of the system (\ref{Eqs:3.11}) is  $1$ when $p \equiv 1 \pmod 4$ and $2p^{m}-1$  otherwise by Lemma 3.5 in \cite{28}.
  \item  The case of $y_{1}$ is a non-square element but  $y_{2}$ is a square element of $\mathbb{F}_{p}^{\ast}$ is equivalent to the above case.
  \item The case of  $y_{1}$ and $y_{2}$ are both non-squares of $\mathbb{F}_{p}^{\ast}$  is equivalent to case (1).
\end{enumerate}
Summarizing the discussion above, the  number of  solutions $(x_{1},x_{2},y_{1},y_{2}) \in \mathbb{F}_{p^{m}}^{2}\times \mathbb{F}_{p}^{\ast2}$ of the system (\ref{Eqs:3.1}) is
$(p-1)^{2}p^{m}$.
\end{proof}
\remark
From the proof of Lemma 3.4 and 3.5 in \cite{28} and the proof of Lemma \ref{Le:3.9}, we have that
the number of solutions $(x_{1},x_{2},y_{1},y_{2}) \in \mathbb{F}_{p^{m}}^{2}\times \mathbb{F}_{p}^{\ast2}$ of (\ref{Eqs:3.1}) is equal to the number of solutions of  the following system of equations
\begin{equation*}
\begin{cases}
y_{1}x_{1}^{2}+ y_{2}x_{2}^{2}=0\\
y_{1}x_{1}^{d_{1}}+y_{2}x_{2}^{d_{1}}=0\\
y_{1}x_{1}^{d_{2}}+y_{2}x_{2}^{d_{2}}=0.
\end{cases}
\end{equation*}
\begin{lemma}\label{Le:3.10}
Let $N_{3}$ denote the number of solutions $(x_{1},x_{2},x_{3},y_{1},y_{2},y_{3}) \in \mathbb{F}_{p^{m}}^{3}\times \mathbb{F}_{p}^{\ast3}$ of the following system of equations
\begin{equation}\label{Eqs:3.2}
\begin{cases}
y_{1}x_{1}^{2}+y_{2}x_{2}^{2}+y_{3} x_{3}^{2}=0\\
y_{1}x_{1}^{d_{1}}+y_{2}x_{2}^{d_{1}}+y_{3} x_{3}^{d_{1}}=0.
\end{cases}
\end{equation}
Then $N_{3}=(p-1)^{3}(p^{m+1}+p^{m}-p)$.
\end{lemma}
\begin{proof}
This lemma can be proved in a similar way as the proof of Lemma \ref{Le:3.9} by Lemma 3.6 and 3.7 in \cite{28}, so we omit the details.
\end{proof}
\remark
Similarly, we also have that
the number of solutions $(x_{1},x_{2},x_{3},y_{1},y_{2},y_{3}) \in \mathbb{F}_{p^{m}}^{3}\times \mathbb{F}_{p}^{\ast3}$ of (\ref{Eqs:3.2}) is equal to the number of solutions of  the following system of equations
\begin{equation*}
\begin{cases}
y_{1}x_{1}^{2}+ y_{2}x_{2}^{2}+ y_{3}x_{3}^{2}=0\\
y_{1}x_{1}^{d_{1}}+y_{2}x_{2}^{d_{1}}+y_{3}x_{3}^{d_{1}}=0\\
y_{1}x_{1}^{d_{2}}+y_{2}x_{2}^{d_{2}}+y_{3}x_{3}^{d_{2}}=0.
\end{cases}
\end{equation*}
\begin{lemma}\label{Le:3.12}
Let $N_{4}$ denote the number of solutions $(x_{1},x_{2},x_{3},x_{4},y_{1},y_{2},y_{3},y_{4}) \in \mathbb{F}_{p^{m}}^{4}\times \mathbb{F}_{p}^{\ast4}$ of the following system of equations
\begin{equation}\label{Eqs:3.3}
\begin{cases}
y_{1}x_{1}^{2}+y_{2}x_{2}^{2}+ y_{3}x_{3}^{2}+y_{4} x_{4}^{2}=0\\
y_{1}x_{1}^{d_{1}}+y_{2}x_{2}^{d_{1}}+y_{3} x_{3}^{d_{1}}+ y_{4}x_{4}^{d_{1}}=0.
\end{cases}
\end{equation}
Then $N_{4}=p^{m}(p^{m+1}+p^{m}-p)(p-1)^{4}$.
\end{lemma}
\begin{proof}
This lemma can be proved in a similar way as the proof of Lemma \ref{Le:3.9} by Lemmas 3.8-3.10 in \cite{28}, so we omit the details.
\end{proof}
\remark
Again, we have that
the number of solutions $(x_{1},x_{2},x_{3},x_{4},y_{1},y_{2},$ $y_{3},y_{4})\in \mathbb{F}_{p^{m}}^{4}\times \mathbb{F}_{p}^{\ast4}$ of (\ref{Eqs:3.3}) is equal to the number of solutions of  the following system of equations
\begin{equation*}
\begin{cases}
y_{1}x_{1}^{2}+ y_{2}x_{2}^{2}+ y_{3}x_{3}^{2}+ y_{4}x_{4}^{2}=0\\
y_{1}x_{1}^{d_{1}}+y_{2}x_{2}^{d_{1}}+y_{3}x_{3}^{d_{1}}+y_{4}x_{4}^{d_{1}}=0\\
y_{1}x_{1}^{d_{2}}+y_{2}x_{2}^{d_{2}}+y_{3}x_{3}^{d_{2}}+y_{4}x_{4}^{d_{2}}=0.
\end{cases}
\end{equation*}
\section{\small{THE CLASS OF REDUCIBLE CODES}}
We follow the notation fixed in the previous sections.
Let $\mathcal{C}_{(p,m,k)}$ be the cyclic code defined in Section 1.
 In this section, we will determine the weight distribution of this class of  cyclic codes. Obviously, $\mathcal{C}_{(p,m,k)}$ has length $p^{m}-1$ and dimension $5m$. Moreover, it can be expressed as \[\mathcal{C}_{(p,m,k)}=\{\mathbf{c}_{(a_{1},a_{2},b_{1},b_{2},c)}: a_{1},a_{2},b_{1},b_{2},c \in \mathbb{F}_{p^{m}}\},\] where \[\mathbf{c}_{(a_{1},a_{2},b_{1},b_{2},c)}=\big(Tr(a_{1}\pi^{t}+a_{2}(-\pi)^{t}+b_{1}\pi^{(p^{k}+1)t/2}+b_{2}(-\pi)^{(p^{k}+1)t/2}+c\pi^{(p^{2k}+1)t/2})\big)_{t=0}^{p^{m}-2}.\]

In terms of exponential sums, the weight of the codeword $\mathbf{c}_{(a_{1},a_{2},b_{1},b_{2},c)}=(c_{0}, c_{1},\ldots,$  $c_{p^{m}-2})$ in $\mathcal{C}_{(p,m,k)}$ is given by
\begin{equation*}
 \begin{split}
 &W(\mathbf{c}_{(a_{1},a_{2},b_{1},b_{2},c)})\\
   &= \#\{0\leq t\leq p^{m}-2: c_{t}\neq 0\}\\
   &= p^{m}-1-\frac{1}{p}\sum_{t=0}^{p^{m}-2}\sum_{y \in \mathbb{F}_{p}}\zeta_{p}^{yc(t)} \\
   &=  p^{m}-1-\frac{1}{p}\sum_{t=0}^{p^{m}-2}\sum_{y \in \mathbb{F}_{p}}\zeta_{p}^{yTr(a_{1}\pi^{t}+a_{2}(-\pi)^{t}+b_{1}\pi^{(p^{k}+1)t/2}+b_{2}(-\pi)^{(p^{k}+1)t/2}+c\pi^{(p^{2k}+1)t/2})} \\
   %&=   p^{m}-1- \frac{1}{p}\sum_{y \in \mathbb{F}_{p}}\sum_{t=0}^{(p^{m}-3)/2}\big(\zeta_{p}^{yTr((a+b)\pi^{2t+1}+c\pi^{(p^{k}+1)t})} +\zeta_{p}^{yTr((a-b)\pi^{2t}+c\pi^{\frac{p^{k}+1}{2}(2t+1)})}\big) \\
   %&=   p^{m}-1- \frac{1}{p}\sum_{y \in \mathbb{F}_{p}}\sum_{x \in SQ}\big(\zeta_{p}^{yTr((a+b)x+cx^{(p^{k}+1)/2})} +\zeta_{p}^{yTr((a-b)\pi x+c(\pi x)^{\frac{p^{k}+1}{2}})}\big) \\
  % &=   p^{m}-1- \frac{1}{p}\sum_{y \in \mathbb{F}_{p}}\sum_{x \in SQ}\big(\zeta_{p}^{yTr((a+b)x+cx^{(p^{k}+1)/2})} +\zeta_{p}^{yTr((a-b)\lambda x+c(\lambda x)^{\frac{p^{k}+1}{2}})}\big) \\
   &=   p^{m}-1- \frac{1}{2p}\sum_{y \in \mathbb{F}_{p}}\sum_{x \in \mathbb{F}_{p^{m}}^{\ast}}\big(\zeta_{p}^{yTr((a_{1}+a_{2})x^{2}+(b_{1}+b_{2})x^{p^{k}+1}+cx^{p^{2k}+1})}\\ &\relphantom{=} \phantom{{} p^{m}-1- \frac{1}{2p}\sum_{y \in \mathbb{F}_{p}}\sum_{x \in \mathbb{F}_{p^{m}}^{\ast}}\big( } +\zeta_{p}^{yTr((a_{1}-a_{2})\lambda x^{2}+(b_{1}-b_{2})\lambda^{\frac{p^{k}+1}{2}} x^{p^{k}+1}+c\lambda^{\frac{p^{2k}+1}{2}} x^{p^{2k}+1})}\big) \\
   &=   p^{m}-p^{m-1}- \frac{1}{2p}\sum_{y \in \mathbb{F}_{p}^{\ast}}\sum_{x \in \mathbb{F}_{p^{m}}}\big(\zeta_{p}^{yTr((a_{1}+a_{2})x^{2}+(b_{1}+b_{2})x^{p^{k}+1}+cx^{p^{2k}+1})}\\ &\relphantom{=} \phantom{{} p^{m}-1- \frac{1}{2p}\sum_{y \in \mathbb{F}_{p}}\sum_{x \in \mathbb{F}_{p^{m}}^{\ast}}\big( } +\zeta_{p}^{yTr((a_{1}-a_{2})\lambda x^{2}+(b_{1}-b_{2})\lambda^{\frac{p^{k}+1}{2}} x^{p^{k}+1}+c\lambda x^{p^{2k}+1})}\big).
 \end{split}
\end{equation*}
It then follows from Lemma \ref{Le:3.1} that $W(\mathbf{c}_{(a_{1},a_{2},b_{1},b_{2},c)})=p^{m}-p^{m-1}- \frac{1}{2p}S(a_{1},a_{2},b_{1},b_{2},c)$ when $k$ is even, where
\begin{equation}\label{Eq:3.1}
 \begin{split}
S(a_{1},a_{2},b_{1},b_{2},c)&=\sum_{y \in \mathbb{F}_{p}^{\ast}}\sum_{x \in \mathbb{F}_{p^{m}}}\big(\zeta_{p}^{yTr((a_{1}+a_{2})x^{2}+(b_{1}+b_{2})x^{p^{k}+1}+cx^{p^{2k}+1})}\\
 &\relphantom{=} \phantom{{} \sum_{y \in \mathbb{F}_{p}^{\ast}}\sum_{x \in \mathbb{F}_{p^{m}}}\big( }+\zeta_{p}^{yTr((a_{1}-a_{2}) x^{2}+(b_{1}-b_{2}) x^{p^{k}+1}+c x^{p^{2k}+1})}\big),
  \end{split}
 \end{equation}
 and $W(\mathbf{c}_{(a_{1},a_{2},b_{1},b_{2},c)})=p^{m}-p^{m-1}- \frac{1}{2p}T(a_{1},a_{2},b_{1},b_{2},c)$ when $k$ is odd, where
 \begin{equation}\label{Eq:3.2}
  \begin{split}
 T(a_{1},a_{2},b_{1},b_{2},c)&=\sum_{y \in \mathbb{F}_{p}^{\ast}}\sum_{x \in \mathbb{F}_{p^{m}}}\big(\zeta_{p}^{yTr((a_{1}+a_{2})x^{2}+(b_{1}+b_{2})x^{p^{k}+1}+cx^{p^{2k}+1})}\\
 &\relphantom{=} \phantom{{} \sum_{y \in \mathbb{F}_{p}^{\ast}}\sum_{x \in \mathbb{F}_{p^{m}}}\big( }+\zeta_{p}^{yTr((a_{1}-a_{2}) x^{2}-(b_{1}-b_{2}) x^{p^{k}+1}+c x^{p^{2k}+1})}\big).
   \end{split}
 \end{equation}
Based on the discussion above, the weight distribution of the code $\mathcal{C}_{(p,m,k)}$ is completely determined by the value distribution of $S(a_{1},a_{2},b_{1},b_{2},c)$ and $T(a_{1},a_{2},b_{1},b_{2},c)$. Before determining the value distribution of $S(a_{1},a_{2},b_{1},b_{2},c)$ and $T(a_{1},a_{2},b_{1},b_{2},c)$,
%Before doing this, we first give a notation.
for any $(u,v,w) \in \mathbb{F}_{p^m}^{2}$, we define
\begin{equation}\label{Eq:3.3}
D(u, v,w)=\sum_{y \in \mathbb{F}_{p}^{\ast}}\sum_{x \in \mathbb{F}_{p^{m}}}\zeta_{p}^{yQ_{u,v,w}(x)}=\sum_{y \in \mathbb{F}_{p}^{\ast}}\sum_{x \in \mathbb{F}_{p^{m}}}\zeta_{p}^{yTr(ux^{2}+vx^{p^{k}+1}+wx^{p^{2k}+1})}.
\end{equation}
From the discussion above, the value  distributions of $S(a_{1},a_{2},b_{1},b_{2},c)$ and $T(a_{1},a_{2},b_{1},b_{2},c)$ can be deduced from the value  distributions of $D(u, v,w)$.
In fact, the value  distributions of $D(u, v,0)$ can be obtained by Lemma 3.2 and 3.3 in \cite{16} as follows.
\begin{lemma}\label{Le:3.3}
Let $D(u,v,w)$ be defined by (\ref{Eq:3.3}). Then as $(u,v)$ runs through $\mathbb{F}_{p^{m}}^{2}$, the value distribution of $D(u,v,0)$ is given by Table \ref{T:3}.
\begin{table}[htbp]
\caption{Value distribution of $D(u,v,0)$}\label{T:3}
\centering
\begin{tabular}{cc}
 \hline
 Value& Frequency\\
 \hline
  0 & $(p^{m}-1)(p^{m}-p^{m-1}+1)$\\
  $(p-1)p^{m}$ & $1$\\
 $(p-1)p^{\frac{m+1}{2}}$ & $\frac{1}{2}(p^{m}-1)(p^{m-1}+p^{\frac{m-1}{2}})$\\
 $-(p-1)p^{\frac{m+1}{2}}$ &  $\frac{1}{2}(p^{m}-1)(p^{m-1}-p^{\frac{m-1}{2}})$\\
 \hline
\end{tabular}
\end{table}
\end{lemma}
For the convenience, we introduce the following notation.
\[
n_{0}=\#\{(u,v)\in \mathbb{F}_{p^{m}}^{2}\mid D(u,v,0)=0\}
\]
and
\[
n_{\epsilon}=\#\{(u,v) \in \mathbb{F}_{p^{m}}^{2}: D(u,v,0)=\varepsilon(p-1)p^{\frac{m+1}{2}}\},
\]
where $\varepsilon=\pm1$. That is $n_{0}=(p^{m}-1)(p^{m}-p^{m-1}+1)$, $n_{\epsilon}=\frac{1}{2}(p^{m}-1)(p^{m-1}+\varepsilon p^{\frac{m-1}{2}})$.
%The following lemmas are very important to establish the value distribution of $S(a_{1},a_{2},b_{1},b_{2},c)$ and $T(a_{1},a_{2},b_{1},b_{2},c)$.
As for any fixed $w \in \mathbb{F}_{p^{m}}^{\ast}$, the value  distributions of $D(u, v,w)$ can be determined by the following lemma.
\begin{lemma}\label{Le:3.4}
Let $D(u,v,w)$ be defined by (\ref{Eq:3.3}). Then for any fixed $w \in \mathbb{F}_{p^{m}}^{\ast}$, as $(u,v)$ runs through $\mathbb{F}_{p^{m}}^{2}$, the value distribution of $D(u,v,w)$ is given by Table  \ref{T:1}.
\begin{table}[htbp]
\caption{Value distribution of $D(u,v,w)$ for any fixed $w \in \mathbb{F}_{p^{m}}^{\ast}$}\label{T:1}
\centering
\begin{tabular}{cc}
 \hline
 Value& Frequency\\
 \hline
  0 & $p^{2m}-p^{2m-1}+p^{2m-4}-p^{m-3}$\\
 $(p-1)p^{\frac{m+1}{2}}$ & $\frac{(p^{m+2}-p^{m}-p^{m-1}+1)(p^{m-1}+p^{\frac{m-1}{2}})}{2(p^{2}-1)}$\\
 $-(p-1)p^{\frac{m+1}{2}}$ & $\frac{(p^{m+2}-p^{m}-p^{m-1}+1)(p^{m-1}-p^{\frac{m-1}{2}})}{2(p^{2}-1)}$\\
 $(p-1)p^{\frac{m+3}{2}}$ &$\frac{(p^{m-1}-1)(p^{m-3}+p^{\frac{m-3}{2}})}{2(p^{2}-1)}$\\
 $-(p-1)p^{\frac{m+3}{2}}$ & $\frac{(p^{m-1}-1)(p^{m-3}-p^{\frac{m-3}{2}})}{2(p^{2}-1)}$\\
 \hline
\end{tabular}
\end{table}
\end{lemma}
\begin{proof}
As in Eq. (\ref{Eq:3.3}),
\begin{equation*}
D(u, v,w)=\sum_{y \in \mathbb{F}_{p}^{\ast}}\sum_{x \in \mathbb{F}_{p^{m}}}\zeta_{p}^{yQ_{u,v,w}(x)}.
\end{equation*} Then for any fixed $w \in \mathbb{F}_{p^{m}}^{\ast}$, by Lemma \ref{Le:2.2} and \ref{Le:2.3},  $D(u, v,w)$ takes on  the values from the set \{0, $\pm(p-1)p^{\frac{m+1}{2}}$, $\pm(p-1)p^{\frac{m+3}{2}}$\}.
To determine the distribution of $D(u,v,w)$ for any fixed $w \in \mathbb{F}_{p^{m}}^{\ast}$, we define
\[
n_{0}^{ w\ast}=\{\#\{(u,v) \in \mathbb{F}_{p^{m}}^{2}: D(u,v,w)=0\}
\]
and
\[
n_{\epsilon,i}^{w}=\#\{(u,v) \in \mathbb{F}_{p^{m}}^{2}: D(u,v,w)=\varepsilon(p-1)p^{\frac{m+i}{2}}\},
\]
where $\varepsilon=\pm1$, $i \in \{1,3\}$. From the following discussion, the number $n_{\epsilon}^{w}$  and $n_{0}^{ w\ast}$ is dependent of the choice of $w \in \mathbb{F}_{p^{m}}^{\ast}$. Hence in the following, for any fixed $w \in \mathbb{F}_{p^{m}}^{\ast}$,  we denote by $n_{\epsilon}$ and $n_{0}^{\ast}$ instead of $n_{\epsilon}^{w}$ and $n_{0}^{ w\ast}$.
Then we have
\begin{equation}\label{Eq:3.4}
\sum_{(u,v) \in \mathbb{F}_{p^{m}}^{2}}D(u,v)=(n_{1,1}-n_{-1,1})(p-1)p^{\frac{m+1}{2}}+(n_{1,3}-n_{-1,3})(p-1)p^{\frac{m+3}{2}},
\end{equation}
\begin{equation}\label{Eq:3.5}
\sum_{(u,v) \in \mathbb{F}_{p^{m}}^{2}}D^{2}(u,v)=(n_{1,1}+n_{-1,1})(p-1)^{2}p^{m+1}+(n_{1,3}+n_{-1,3})(p-1)^{2}p^{m+3},
\end{equation}
\begin{equation}\label{Eq:3.8}
\sum_{(u,v) \in \mathbb{F}_{p^{m}}^{2}}D(u,v)=(n_{1,1}-n_{-1,1})(p-1)^{3}p^{\frac{3(m+1)}{2}}+(n_{1,3}-n_{-1,3})(p-1)^{3}p^{\frac{3(m+3)}{2}}
\end{equation}
and
\begin{equation}\label{Eq:3.9}
\sum_{(u,v) \in \mathbb{F}_{p^{m}}^{2}}D^{2}(u,v)=(n_{1,1}+n_{-1,1})(p-1)^{4}p^{2(m+1)}+(n_{1,3}+n_{-1,3})(p-1)^{4}p^{2(m+3)}.
\end{equation}
On the other hand, it follows from (\ref{Eq:3.3}) that
\begin{equation}\label{Eq:3.6}
\begin{split}
\sum_{u,v \in \mathbb{F}_{p^{m}}}D(u,v,w)&=\sum_{u,v \in \mathbb{F}_{p^{m}}}\sum_{y \in \mathbb{F}_{p}^{\ast}}\sum_{x \in \mathbb{F}_{p^{m}}}\zeta_{p}^{yTr(ux^{2}+vx^{p^{k}+1}+wx^{p^{2k}+1})}\\
&=\sum_{y \in \mathbb{F}_{p}^{\ast}}\sum_{x \in \mathbb{F}_{p^{m}}}\zeta_{p}^{yTr(wx^{p^{2k}+1})}\sum_{u \in \mathbb{F}_{p^{m}}}\zeta_{p}^{yTr(ux^{2})}\sum_{v \in \mathbb{F}_{p^{m}}}\zeta_{p}^{yTr(vx^{p^{k}+1})}\\
&=(p-1)p^{2m}.
 \end{split}
\end{equation}
By Lemmas \ref{Le:3.9}-\ref{Le:3.12} and the remarks followed, we have the following results.
\begin{equation}\label{Eq:3.7}
\begin{split}
&\sum_{u,v \in \mathbb{F}_{p^{m}}}D^{2}(u,v,w)\\
&=\sum_{u,v \in \mathbb{F}_{p^{m}}}\sum_{y_{1},y_{2} \in \mathbb{F}_{p}^{\ast}}\sum_{x_{1},x_{2} \in \mathbb{F}_{p^{m}}}\zeta_{p}^{y_{1}Tr(ux_{1}^{2}+vx_{1}^{p^{k}+1}+wx_{1}^{p^{2k}+1})}
\zeta_{p}^{y_{2}Tr(ux_{2}^{2}+vx_{2}^{p^{k}+1}+wx_{2}^{p^{2k}+1})}\\
&=\sum_{(y_{1},y_{2}) \in \mathbb{F}_{p}^{\ast 2}}\sum_{(x_{1},x_{2}) \in \mathbb{F}_{p^{m}}^{2}}\zeta_{p}^{Tr(w(y_{1}x_{1}^{p^{2k}+1}+y_{2}x_{2}^{p^{2k}+1}))}\cdot\\
&\relphantom{=}{}\sum_{u \in \mathbb{F}_{p^{m}}}\zeta_{p}^{Tr(u(y_{1}x_{1}^{2}+y_{2}x_{2}^{2}))}\sum_{v \in \mathbb{F}_{p^{m}}}\zeta_{p}^{Tr(v(y_{1}x_{1}^{p^{k}+1}+y_{2}x_{2}^{p^{k}+1}))}\\
&=p^{2m}\cdot \#\{(x_{1},x_{2},y_{1},y_{2}) \in \mathbb{F}_{p^{m}}^{2}\times \mathbb{F}_{p}^{\ast2}\mid y_{1}x_{1}^{2}+ y_{2}x_{2}^{2}=0,
y_{1}x_{1}^{d_{1}}+y_{2}x_{2}^{d_{1}}=0\}\\
&=(p-1)^{2}p^{2m},
\end{split}
\end{equation}
\begin{equation}\label{Eq:3.10}
\begin{split}
&\sum_{u,v \in \mathbb{F}_{p^{m}}}D^{3}(u,v,w)\\
&=\sum_{(y_{1},y_{2},y_{3}) \in \mathbb{F}_{p}^{\ast 2}}\sum_{(x_{1},x_{2},x_{3}) \in \mathbb{F}_{p^{m}}^{2}}\zeta_{p}^{Tr(w(y_{1}x_{1}^{p^{2k}+1}+y_{2}x_{2}^{p^{2k}+1}+y_{3}x_{3}^{p^{2k}+1}))}\cdot\\
&\relphantom{=}{}\sum_{u \in \mathbb{F}_{p^{m}}}\zeta_{p}^{Tr(u(y_{1}x_{1}^{2}+y_{2}x_{2}^{2}+y_{3}x_{3}^{2}))}\sum_{v \in \mathbb{F}_{p^{m}}}\zeta_{p}^{Tr(v(y_{1}x_{1}^{p^{k}+1}+y_{2}x_{2}^{p^{k}+1}+y_{3}x_{3}^{p^{k}+1}))}\\
&=p^{2m}\cdot \#\{(x_{1},x_{2},x_{3},y_{1},y_{2},y_{3}) \in \mathbb{F}_{p^{m}}^{3}\times \mathbb{F}_{p}^{\ast3}\mid y_{1}x_{1}^{2}+y_{2}x_{2}^{2}+y_{3} x_{3}^{2}=0,\\
&\relphantom{=}{}y_{1}x_{1}^{d_{1}}+y_{2}x_{2}^{d_{1}}+y_{3} x_{3}^{d_{1}}=0.\}\\
&=p^{2m}(p-1)^{3}(p^{m+1}+p^{m}-p)
\end{split}
\end{equation}
and
\begin{equation}\label{Eq:3.11}
\begin{split}
&\sum_{u,v \in \mathbb{F}_{p^{m}}}D^{4}(u,v,w)\\
&=\sum_{(y_{1},y_{2},y_{3},y_{4}) \in \mathbb{F}_{p}^{\ast 2}}\sum_{(x_{1},x_{2},x_{3},x_{4}) \in \mathbb{F}_{p^{m}}^{2}}\zeta_{p}^{Tr(w(y_{1}x_{1}^{p^{2k}+1}+y_{2}x_{2}^{p^{2k}+1}+y_{3}x_{3}^{p^{2k}+1}+y_{4}x_{4}^{p^{2k}+1}))}\cdot\\
&\relphantom{=}{}\sum_{u \in \mathbb{F}_{p^{m}}}\zeta_{p}^{Tr(u(y_{1}x_{1}^{2}+y_{2}x_{2}^{2}+y_{3}x_{3}^{2}+y_{4}x_{4}^{2}))}\sum_{v \in \mathbb{F}_{p^{m}}}\zeta_{p}^{Tr(v(y_{1}x_{1}^{p^{k}+1}+y_{2}x_{2}^{p^{k}+1}+y_{3}x_{3}^{p^{k}+1}+y_{4}x_{4}^{p^{k}+1}))}\\
&=p^{2m}\cdot \#\{(x_{1},x_{2},x_{3},x_{4},y_{1},y_{2},y_{3},y_{4})\in \mathbb{F}_{p^{m}}^{4}\times \mathbb{F}_{p}^{\ast4}\mid y_{1}x_{1}^{2}+y_{2}x_{2}^{2}+ y_{3}x_{3}^{2}+y_{4} x_{4}^{2}=0,\\
&\relphantom{=}{}
y_{1}x_{1}^{d_{1}}+y_{2}x_{2}^{d_{1}}+y_{3} x_{3}^{d_{1}}+ y_{4}x_{4}^{d_{1}}=0.\}\\
&=p^{3m}(p-1)^{4}(p^{m+1}+p^{m}-p).
\end{split}
\end{equation}
Combining Eqs. (\ref{Eq:3.4})-(\ref{Eq:3.11}), we have
\[
\begin{split}
&n_{1,1}=\frac{(p^{m+2}-p^{m}-p^{m-1}+1)(p^{m-1}+p^{\frac{m-1}{2}})}{2(p^{2}-1)},\\
&n_{-1,1}=\frac{(p^{m+2}-p^{m}-p^{m-1}+1)(p^{m-1}-p^{\frac{m-1}{2}})}{2(p^{2}-1)},\\
&n_{1,3}=\frac{(p^{m-1}-1)(p^{m-3}+p^{\frac{m-3}{2}})}{2(p^{2}-1)},\\
&n_{-1,3}=\frac{(p^{m-1}-1)(p^{m-3}-p^{\frac{m-3}{2}})}{2(p^{2}-1)}.
\end{split}
\]
Then we have $n_{0}^{\ast}=p^{2m}-n_{1,1}-n_{-1,1}-n_{1,3}-n_{-1,3}=p^{2m}-p^{2m-1}+p^{2m-4}-p^{m-3}$.
\end{proof}
The value distribution of $S(a_{1},a_{2},b_{1},b_{2},c)$ is determined by the following theorem.

\begin{theorem}\label{Th:3.5}
Let $k$ be even and  $S(a_{1},a_{2},b_{1},b_{2},c)$ be defined by (\ref{Eq:3.1}). Then as $(a_{1},a_{2},b_{1},$
$b_{2},c)$ runs through $\mathbb{F}_{p^{m}}^{5}$, the value distribution of $S(a_{1},a_{2},b_{1},b_{2},c)$ is given by Table \ref{T:8}.
\end{theorem}
\newpage
\pagestyle{plain}
\begin{table}[htbp]
\caption{Value Distribution of $S(a_{1},a_{2},b_{1},b_{2},c)$}\label{T:8}
\centering
\begin{tabular}{ll}
 \hline
 Value& Frequency\\
 \hline
 $2(p-1)p^{m}$ &1\\
 $(p-1)p^{m}$ &  $2n_{0}$\\
  $(p-1)(p^{m}-p^{\frac{m+1}{2}})$ &  $2n_{-1}$\\
  $(p-1)(p^{m}+p^{\frac{m+1}{2}})$ &  $2n_{1}$\\
 $(p-1)p^{\frac{m+1}{2}}$ & $2n_{0}n_{1}+2(p^{m}-1)n_{0}^{\ast}n_{1,1}$\\
  $(p-1)p^{\frac{m+3}{2}}$ & $2(p^{m}-1)n_{0}^{\ast}n_{1,3}$\\
  $2(p-1)p^{\frac{m+1}{2}}$ & $n_{1}^{2}+(p^{m}-1)n_{1,1}^{2}$\\
  $2(p-1)p^{\frac{m+3}{2}}$ & $(p^{m}-1)n_{1,3}^{2}$\\
 $-(p-1)p^{\frac{m+1}{2}}$ & $2n_{0}n_{-1}+2(p^{m}-1)n_{0}^{\ast}n_{-1,1}$\\

 $-(p-1)p^{\frac{m+3}{2}}$ & $2(p^{m}-1)n_{0}^{\ast}n_{-1,3}$\\
 $-2(p-1)p^{\frac{m+1}{2}}$ & $n_{-1}^{2}+(p^{m}-1)n_{-1,1}^{2}$\\
 $-2(p-1)p^{\frac{m+3}{2}}$ & $(p^{m}-1)n_{-1,3}^{2}$\\
 $(p-1)(p^{\frac{m+1}{2}}+p^{\frac{m+3}{2}})$ &  $2(p^{m}-1)n_{1,1}n_{1,3}$\\
 $(p-1)(p^{\frac{m+1}{2}}-p^{\frac{m+3}{2}})$ &  $2(p^{m}-1)n_{1,1}n_{-1,3}$\\
 $(p-1)(-p^{\frac{m+1}{2}}+p^{\frac{m+3}{2}})$ &  $2(p^{m}-1)n_{-1,1}n_{1,3}$\\
 $(p-1)(-p^{\frac{m+1}{2}}-p^{\frac{m+3}{2}})$ &  $2(p^{m}-1)n_{-1,1}n_{-1,3}$\\
 $0$ & \tabincell{l}{$n_{0}^{2}+2n_{1}n_{-1}+(p^{m}-1)n_{0}^{\ast2}$\\$+2(p^{m}-1)(n_{1,1}n_{-1,1}+n_{1,3}n_{-1,3})$}\\
 \hline
\end{tabular}
\end{table}
\begin{proof}
By Lemma \ref{Le:3.3} and \ref{Le:3.4}, $S(a_{1},a_{2},b_{1},b_{2},c)$ takes values from the set $\{0, (p-1)p^{m}, 2(p-1)p^{m}, \pm(p-1)p^{\frac{m+1}{2}}, (p-1)(p^{m}\pm p^{\frac{m+1}{2}}), \pm2(p-1)p^{\frac{m+1}{2}}, \pm(p-1)p^{\frac{m+3}{2}}, \pm2(p-1)p^{\frac{m+3}{2}},(p-1)(p^{\frac{m+1}{2}}\pm p^{\frac{m+3}{2}}), (p-1)(-p^{\frac{m+1}{2}}\pm p^{\frac{m+3}{2}})\}$. The distribution of $S(a,b,c)=(p-1)p^{m}$, $2(p-1)p^{m}$ or $(p-1)(p^{m}\pm p^{\frac{m+1}{2}})$ can be easily obtained by Lemma \ref{Le:3.3}.
\[
\begin{split}
&\#\{(a_{1},a_{2},b_{1},b_{2},c) \in \mathbb{F}_{p^{m}}^{5}: S(a_{1},a_{2},b_{1},b_{2},c)=0\}\\
&=\#\{(u_{1},u_{2},v_{1},v_{2},c) \in \mathbb{F}_{p^{m}}^{5}: D(u_{1},v_{1},c)+D(u_{2},v_{2},c)=0\}\\
&=\#\{(u_{1},u_{2},v_{1},v_{2}) \in \mathbb{F}_{p^{m}}^{4}: D(u_{1},v_{1},0)+D(u_{2},v_{2},0)=0\}\\
&\relphantom{=} {}+\#\{(u_{1},u_{2},v_{1},v_{2}) \in \mathbb{F}_{p^{m}}^{4},c\in \mathbb{F}_{p^{m}}^{\ast}: D(u_{1},v_{1},c)+D(u_{2},v_{2},c)=0\}\\
&=n_{0}^{2}+2n_{1}n_{-1}+(p^{m}-1)n_{0}^{\ast2}+2(p^{m}-1)(n_{1,1}n_{-1,1}+n_{1,3}n_{-1,3}).
\end{split}
\]
Similarly, we have
\[
\begin{split}
&\#\{(a_{1},a_{2},b_{1},b_{2},c) \in \mathbb{F}_{p^{m}}^{5}: S(a_{1},a_{2},b_{1},b_{2},c)=(p-1)p^{\frac{m+1}{2}}\}=2n_{0}n_{1}+2(p^{m}-1)n_{0}^{\ast}n_{1,1},\\
&\#\{(a_{1},a_{2},b_{1},b_{2},c) \in \mathbb{F}_{p^{m}}^{5}: S(a_{1},a_{2},b_{1},b_{2},c)=-(p-1)p^{\frac{m+1}{2}}\}=2n_{0}n_{-1},\\
&\relphantom{=} {}+2(p^{m}-1)n_{0}^{\ast}n_{-1,1},\\
&\#\{(a_{1},a_{2},b_{1},b_{2},c) \in \mathbb{F}_{p^{m}}^{5}: S(a_{1},a_{2},b_{1},b_{2},c)=2(p-1)p^{\frac{m+1}{2}}\}=n_{1}^{2}+(p^{m}-1)n_{1,1}^{2},\\
&\#\{(a_{1},a_{2},b_{1},b_{2},c) \in \mathbb{F}_{p^{m}}^{5}: S(a_{1},a_{2},b_{1},b_{2},c)=2(p-1)p^{\frac{m+1}{2}}\}=n_{-1}^{2}+(p^{m}-1)n_{-1,1}^{2},\\
\end{split}
\]
\[
\begin{split}
&\#\{(a_{1},a_{2},b_{1},b_{2},c) \in \mathbb{F}_{p^{m}}^{5}: S(a_{1},a_{2},b_{1},b_{2},c)=(p-1)p^{\frac{m+3}{2}}\}=2(p^{m}-1)n_{0}^{\ast}n_{1,3},\\
&\#\{(a_{1},a_{2},b_{1},b_{2},c) \in \mathbb{F}_{p^{m}}^{5}: S(a_{1},a_{2},b_{1},b_{2},c)=-(p-1)p^{\frac{m+3}{2}}\}=2(p^{m}-1)n_{0}^{\ast}n_{-1,3},\\
&\#\{(a_{1},a_{2},b_{1},b_{2},c) \in \mathbb{F}_{p^{m}}^{5}: S(a_{1},a_{2},b_{1},b_{2},c)=2(p-1)p^{\frac{m+3}{2}}\}=(p^{m}-1)n_{1,3}^{2},\\
&\#\{(a_{1},a_{2},b_{1},b_{2},c) \in \mathbb{F}_{p^{m}}^{5}: S(a_{1},a_{2},b_{1},b_{2},c)=-2(p-1)p^{\frac{m+3}{2}}\}=(p^{m}-1)n_{-1,3}^{2},\\
&\#\{(a_{1},a_{2},b_{1},b_{2},c) \in \mathbb{F}_{p^{m}}^{5}: S(a_{1},a_{2},b_{1},b_{2},c)=(p-1)(p^{\frac{m+1}{2}}+p^{\frac{m+3}{2}})\\
&=2(p^{m}-1)n_{1,1}n_{1,3},\\
&\#\{(a_{1},a_{2},b_{1},b_{2},c) \in \mathbb{F}_{p^{m}}^{5}: S(a_{1},a_{2},b_{1},b_{2},c)=(p-1)(p^{\frac{m+1}{2}}-p^{\frac{m+3}{2}})\}\\
&=2(p^{m}-1)n_{1,1}n_{-1,3},\\
&\#\{(a_{1},a_{2},b_{1},b_{2},c) \in \mathbb{F}_{p^{m}}^{5}: S(a_{1},a_{2},b_{1},b_{2},c)=(p-1)(-p^{\frac{m+1}{2}}+p^{\frac{m+3}{2}})\}\\
&=2(p^{m}-1)n_{-1,1}n_{1,3},\\
&\#\{(a_{1},a_{2},b_{1},b_{2},c) \in \mathbb{F}_{p^{m}}^{5}: S(a_{1},a_{2},b_{1},b_{2},c)=(p-1)(-p^{\frac{m+1}{2}}-p^{\frac{m+3}{2}})\}\\
&=2(p^{m}-1)n_{-1,1}n_{-1,3}.
\end{split}
\]
\end{proof}
By calculation, the sum of  all the numbers in the right side of Table \ref{T:8}  is exactly $p^{5m}$, which can be used to  verify the correctness of the theorem above.
\remark  Following the notation above, we have
$T(a_{1},a_{2},b_{1},b_{2},c)=D(a_{1}+a_{2},b_{1}+b_{2},c)+D(a_{1}+a_{2},-(b_{1}-b_{2}),c)$. It can be shown that  the value distribution of $T(a_{1},a_{2},$
$b_{1},b_{2},c)$ in the case of $k$ is odd is the same as the value distribution of $S(a_{1},a_{2},b_{1},b_{2},c)$ in the case of $k$ is even.

\noindent The following is the main result of this paper.
\begin{theorem}\label{Th:3.6}
Let $m$ and $k$ be any two  positive integers such that $m \geq 5$ is odd, then
$\mathcal{C}_{(p,m,k)}$ is a  cyclic code over $\mathbb{F}_{p}$ with parameters $[p^{m}-1, 5m, \frac{1}{2}(p-1)(p^{m-1}-p^{\frac{m+1}{2}})]$. Furthermore, the weight distribution of $\mathcal{C}_{(p,m,k)}$ is given by Table \ref{T:2}.
\end{theorem}
\begin{proof}
The length and dimension of $\mathcal{C}_{(p,m,k)}$ follow directly from its definition. The minimal weight and weight distribution of $\mathcal{C}_{(p,m,k)}$ follow from Eqs. (\ref{Eq:3.1}) and  (\ref{Eq:3.2}), Theorem \ref{Th:3.5} and the Remark above.
\end{proof}
\remark The weight distribution of $\mathcal{C}_{(p,m,k)}$ can be determined in a similar way in a more general case where $m/ gcd(m,k)\geq 5$ is odd. We omit the details in order to
avoid duplication.

\newpage
\pagestyle{plain}
\begin{table}[htbp]
\caption{Weight Distribution of $\mathcal{C}_{(p,m,k)}$}\label{T:2}
\centering
\begin{tabular}{ll}
 \hline
 Weight& Frequency\\
 \hline
  $0$ &1\\
 $\frac{1}{2}(p-1)p^{m-1}$ &  $2n_{0}$\\
  $\frac{1}{2}(p-1)(p^{m-1}+p^{\frac{m-1}{2}})$ &  $2n_{-1}$\\
  $\frac{1}{2}(p-1)(p^{m-1}-p^{\frac{m+1}{2}})$ &  $2n_{1}$\\
 $(p-1)(p^{m-1}-\frac{1}{2}p^{\frac{m-1}{2}})$ & $2n_{0}n_{1}+2(p^{m}-1)n_{0}^{\ast}n_{1,1}$\\
  $(p-1)(p^{m-1}-\frac{1}{2}p^{\frac{m+1}{2}})$ & $2(p^{m}-1)n_{0}^{\ast}n_{1,3}$\\
  $(p-1)(p^{m-1}-p^{\frac{m-1}{2}})$ & $n_{1}^{2}+(p^{m}-1)n_{1,1}^{2}$\\
  $(p-1)(p^{m-1}-p^{\frac{m+1}{2}})$ & $(p^{m}-1)n_{1,3}^{2}$\\
 $(p-1)(p^{m-1}+\frac{1}{2}p^{\frac{m-1}{2}})$ & $2n_{0}n_{-1}+2(p^{m}-1)n_{0}^{\ast}n_{-1,1}$\\

 $(p-1)(p^{m-1}+\frac{1}{2}p^{\frac{m+1}{2}})$& $2(p^{m}-1)n_{0}^{\ast}n_{-1,3}$\\
 $(p-1)(p^{m-1}+p^{\frac{m-1}{2}})$ & $n_{-1}^{2}+(p^{m}-1)n_{-1,1}^{2}$\\
 $(p-1)(p^{m-1}-p^{\frac{m+1}{2}})$ & $(p^{m}-1)n_{-1,3}^{2}$\\
$(p-1)(p^{m-1}-\frac{1}{2}p^{\frac{m-1}{2}}-\frac{1}{2}p^{\frac{m+1}{2}})$ &  $2(p^{m}-1)n_{1,1}n_{1,3}$\\
 $(p-1)(p^{m-1}-\frac{1}{2}p^{\frac{m-1}{2}}+\frac{1}{2}p^{\frac{m+1}{2}})$ &  $2(p^{m}-1)n_{1,1}n_{-1,3}$\\
 $(p-1)(p^{m-1}+\frac{1}{2}p^{\frac{m-1}{2}}-\frac{1}{2}p^{\frac{m+1}{2}})$ &  $2(p^{m}-1)n_{-1,1}n_{1,3}$\\
 $(p-1)(p^{m-1}+\frac{1}{2}p^{\frac{m-1}{2}}+\frac{1}{2}p^{\frac{m+1}{2}})$ &  $2(p^{m}-1)n_{-1,1}n_{-1,3}$\\
 $(p-1)p^{m-1}$ & \tabincell{l}{$n_{0}^{2}+2n_{1}n_{-1}+(p^{m}-1)n_{0}^{\ast2}$\\$+2(p^{m}-1)(n_{1,1}n_{-1,1}+n_{1,3}n_{-1,3})$}\\
 \hline
\end{tabular}
\end{table}

\end{document}